\documentclass[1p]{elsarticle}

\usepackage[T2A]{fontenc}
\usepackage[utf8]{inputenc}
\usepackage[english]{babel}
\usepackage{amssymb,amsfonts,amsmath,amsthm}
\usepackage{xcolor}

\journal{Journal of \LaTeX\ Templates}

\newtheorem{theorem}{Theorem}
\newtheorem{consequence}{Consequence}

\newcommand{\p}{\partial}
\newcommand{\g}{\mathfrak{g}}
\newcommand{\G}{\mathbf{G}}

\renewcommand{\O}{{\cal O}}

\newcommand{\eqdef}{\stackrel{\mathrm{def}}{=}}

\begin{document}
\selectlanguage{english}

\begin{frontmatter}

\title{Constructing a complete integral of the Hamilton--Jacobi equation on pseudo-Riemannian spaces with simply transitive groups of motions}

\author{A.~A.~Magazev\fnref{myfootnote}}
\address{Omsk State Technical University, Prospect Mira, 11, Omsk, Russia, 644050}
\ead{magazev@omgtu.ru}

\begin{abstract}
In this work, an efficient method for constructing a complete integral of the geodesic Hamilton–Jacobi equation on pseudo-Riemannian manifolds with simply tran\-sitive groups of motions is suggested. The method is based on using a special transition to canonical coordinates on coadjoint orbits of the group of motion. As a non-trivial example, we consider the problem of constructing a complete integral of the geodesic Hamilton–Jacobi equation in the McLenaghan--Tariq--Tupper spacetime. An essential feature of this example is that the Hamilton-Jacobi equation is not separable in the corresponding configuration space.
\end{abstract}


\end{frontmatter}


\section{Introduction}
\label{intro}
Let $(M, g)$ be a pseudo-Riemannian manifold of dimension $n$ and let $x = ( x^1, \dots, x^n )$ be a local system of coordinates on~$M$. 
The \textit{geodesic Hamilton--Jacobi equation} on $(M, g)$ is the first-order partial differential equation  
\begin{equation}
\label{eq:01}
g^{ij} \frac{\p S}{\p x^i} \frac{\p S}{\p x^j} = m^2,
\end{equation}
where $g^{ij}$ are the contravariant components of the metric, $m$ is a real non-negative parameter. Hereafter, we assume a summation over the repeated indices (Einstein notation). 
A \textit{complete integral} of this equation is a solution $S(x;\alpha)$ depending on $n$ real parameters $\alpha = (\alpha_1, \dots, \alpha_n)$ such that
\begin{equation}
\label{eq:02}
\det \left \| \frac{\p^2 S(x;\alpha)}{\p x^i\, \p \alpha_j} \right \| \neq 0.
\end{equation}

The problem of finding a complete integral of the Hamilton--Jacobi equation is closely related to the problem of solving the geodesic equations (see, for example,~\cite{Eis66}). 
Indeed, if we know a complete integral of \eqref{eq:01}, then, by purely algebraic manipulations, we can construct integral curves of the Hamiltonian system with the \textit{geodesic Hamiltonian} $H(x,p) = \frac{1}{2}\, g^{ij} p_i p_j$. 
This method of constructing solutions of the geodesic equations, called the \textit{Jacobi method}, is often used in analytical mechanics and general relativity~\cite{GolPooSaf00,Arn89}.

The most traditional approach to finding a complete integral for the Ha\-mil\-ton--Jacobi equation is based on the concept of \textit{separation of variables}. 
Although the theory of separation of variables has a long history going back to works of Jacobi, Levi-Civita, St\"ackel, and Eisenhart, modern physicists and mathematicians still demonstrate their interest in this field. 
We will not give a review of this theory in this paper; instead, we refer the reader to the works \cite{Woo75,Sha79,KalMil80,BagObu93,BenChaRas01,Ben16} where the problem of separation of variables in the geodesic Hamilton--Jacobi equation has been discussed in detail. 
Here we remark only that, as shown by Eisenhart \cite{Eis66,Eis34} and by Kalnins and Miller~\cite{KalMil80,KalMil81}, the geodesic separation implies the existence of Killing vectors and Killing tensors of order two. 
Note, however, that there are the so-called \textit{non-St\"ackel} pseudo-Riemannian manifolds for which \eqref{eq:01} is non-separable on the configuration space (some examples of such four-dimensional manifolds can be found in \cite{Kli00}).   
Thus, the development of new alternative approaches to constructing complete integrals of the geodesic Hamilton-Jacobi is also important for physical and geometric applications.

There are different ways to go beyond the framework of the method of separation of variables.
For instance, instead of coordinate transformations in~$M$, one can consider more general canonical transformations of the whole phase space $T^* M$. First interest examples of such separation were pointed out in rigid body dynamics (for instance, the Kovalevskaya top and the Steklov--Lyapunov systems). Other examples of separable systems which have non-St\"ackel form can be found in Babelon \textit{et al}.~\cite{BabBerTal03}. Unfortunately, there is no universal method for separating variables in such problems, and the techniques used here are rather nontrivial.

Another alternative to the method of separation of variables is the \textit{symplectic reduction} procedure~\cite{MarWei74,MarAbr78}.
The advantage of this approach is the possibility of integrating the geodesic equations without involving ``hidden symmetries'' related to Killing tensors of order two (and higher orders).
Recall that the main idea of the symplectic reduction is to apply the noncommutative symmetries associated with conservation laws to decrease the number of phase variables of a mechanical system required for describing its dynamics. 
The difficulty, however, is the fact that this procedure is directly applied to Hamiltonian systems, while the Hamilton--Jacobi equation is a partial differential equation. 
Nevertheless, there is some success in combining both symplectic reduction theory and Hamilton–Jacoby theory. In particular, H.~Wang has proved a number of Hamilton--Jacobi theorems for regular reducible Hamiltonian systems on cotangent bundles and employed these results to the investigation of some rigid body and heavy top systems~\cite{Wan17}. The more general results were given by M.~de Le\'on \textit{et al}.~\cite{LeoDieVaq17} and are related to the reduction and reconstruction procedures for the Hamilton--Jacobi equation with symmetries. 
However, these authors restricted themselves only $G$-invariant solutions of the Hamilton--Jacobi equation and did not practically consider the problem of constructing its complete solution.

In this paper, we solve the problem of constructing a complete integral for the geodesic Hamilton--Jacobi equation on a pseudo-Riemannian manifold $(M, g)$ admitting a simply transitive isometry group~$G$. 
Following the basic idea of symplectic reduction, we reduce this problem to an auxiliary first-order partial differential equation on Lagrangian submanifolds of regular coadjoint orbits of the group $G$. 
This auxiliary equation involves fewer independent variables than the original one and, in some cases, can be solved by quadratures.
An essential feature of our method is an explicit formula reconstructing a complete integral of the original Hamilton--Jacobi equation from a complete integral of the reduced equation (see the formula~\eqref{eq:30}). It should be also emphasized that the technique for constructing a complete integral of \eqref{eq:01} developed in the paper is fully constructive: all steps of the technique are reduced to using only quadratures and tools of linear algebra.

The paper is organized as follows. In Sec.~\ref{sec:1}, we recall some facts concerning pseudo-Riemannian manifolds with simply transitive groups of motions.
In particular, we write the Hamilton--Jacobi equation \eqref{eq:01} in terms of invariant vector fields on $(M,g)$. 
Sec.~\ref{sec:2} is devoted to coadjoint orbits of Lie groups, especially the problem of constructing canonical coordinates (the Darboux coordinates) on the orbits.
We pay special attention to a specific class of canonical coordinates that are related to polarizations of Lie algebras. 
In Sec.~\ref{sec:3}, we prove a theorem containing the main result of the present paper. 
According to this theorem, a complete integral of ~\eqref{eq:01} can be constructed by solving an auxiliary Hamilton--Jacobi equation on the invariant Lagrangian submanifolds of regular coadjoint orbits of $G$, where $G$ is the isometry group of the pseudo-Riemannian manifold~$(M, g)$. 
For the Lie groups with two-dimensional regular orbits, this result allows us to integrate \eqref{eq:01} by quadratures alone.

In the final section, as an application of our method, we consider the problem of constructing a complete integral of the geodesic Hamilton--Jacobi equation in the \textit{McLenaghan--Tariq--Tapper spacetime}. 
The metric of this spacetime was found by McLenaghan and Tariq~\cite{LenTar75} and Tapper \cite{TarTap75} as an electrovacuum solution whose electromagnetic tensor does not share the spacetime symmetry.
It is noteworthy that the geodesic Hamilton--Jacobi equation cannot be solved by the separation of variables for this spacetime.
Nevertheless, employing our method, we construct a complete integral of \eqref{eq:01} in terms of incomplete elliptic integrals.

\section{The Hamilton--Jacobi equation in terms of invariant vector fields}
\label{sec:1}

We start with recalling some facts concerning pseudo-Riemannian manifolds with simply transitive groups of motions. 
For additional details and references, we refer the reader to the book of Stephani \textit{et al.}~\cite{Ste03}.

Let $(M,g)$ be a differentiable pseudo-Riemannian manifold admitting an $n$-dimensional group of motions~$G$:
$$
\tau_z^* g = g,\quad z \in G.
$$
Here $\tau_z: M \to M$ is the transformation associated with the group element~$z$. 
Everywhere in this paper, we assume that the group $G$ acts on $M$ \textit{on the left}, that is
$$
\tau_{z_1 z_2} x = \tau_{z_1} \tau_{z_2} x,\quad
x \in M,\ z_1, z_2 \in G.
$$

Suppose that the action of $G$ on $M$ is \textit{simply transitive}. 
This means that for any two points $x_1, x_2 \in M$ there exists one and only one group element $z \in G$ such that $x_2 = \tau_z x_1$. 
Choose a point $x_0 \in M$ and define a mapping $\psi_{x_0}: G \to M$ by
\begin{equation}
\label{eq:05}
\psi_{x_0}(z) \eqdef \tau_z\, x_0,\quad
z \in G.
\end{equation}
It is clear that this mapping establishes a smooth one-to-one correspondence between points of~$M$ and elements of $G$. 
In particular, the dimensions of manifold $M$ and the Lie group $G$ coincide: $\dim M = \dim G = n$.

We recall that the \textit{left translation} associated with $z \in G$ is a mapping $L_z: G \to G$ such that $L_z(z') = z z'$ for all $z' \in G$. 
From (\ref{eq:05}) it follows that 
\begin{equation}
\label{eq:06}
\psi_{x_0} \circ L_z = \tau_z \circ \psi_{x_0},\quad
z \in G,
\end{equation}
i.e. the mapping $\psi_{x_0}$ is equivariant with respect to the actions of $G$ on $G$ and $M$, respectively. 
In analogy with the left translation, one can define the \textit{right translation} $R_z: G \to G$ by $R_z(z') = z z'$. 
As can be readily seen, the right translations commute with the left ones, $R_z \circ L_{z'} = L_{z'} \circ R_z$, $z, z' \in G$.

Let $\g = T_e G$ be the Lie algebra of the group $G$ and let $e_1, \dots, e_n$ be a basis in~$\g$. 
The Lie bracket $[\cdot,\cdot]$ on~$\g$ is completely determined by the commutation relations between its basis elements: 
\begin{equation*}
\label{eq:07}
[e_i, e_j] = C_{ij}^k\, e_k.
\end{equation*}
The coefficients $C_{ij}^k$ are known as the \textit{structure constants} of the group~$G$. 

Consider the left-invariant vector field $l_i(z) = (L_z)_*\,e_i$ on $G$ corresponding to the basis vector~$e_i \in \g$ and denote by $\xi_i$ the image of $l_i$ under the mapping~$\psi_{x_0}$: 
\begin{equation}
\label{eq:08}
\xi_i(\tau_z\, x_0) \eqdef (\psi_{x_0})_*\, l_i(z),\quad
z \in G.
\end{equation}
By virtue of \eqref{eq:06}, the vector filed $\xi_i$ is invariant with respect to the action of~$G$ on~$M$. 
Moreover, it is easy to show that 
$$
[\xi_i, \xi_j] = C_{ij}^k\, \xi_k,
$$
i.~e. the vector fields $\xi_i$ span a Lie algebra $\g_L(M)$ isomorphic to the Lie algebra~$\g$.

Similarly, we can define the vector fields $\eta_i$ on $M$ by the formula
$$
\eta_i(\tau_z\, x_0) \eqdef (\psi_{x_0})_*\, r_i(z),\quad
z \in G,
$$
where $r_i(z) = (R_z)_*\, e_i$ is the right-invariant vector field on $G$ associated with $e_i \in \g$. 
Unlike the vector fields $\xi_i$, the fields $\eta_i$ are not invariant under the group action $G$ on $M$. 
Furthermore, these vector fields commute with $\xi_i$ and form a Lie algebra $\g_R(M)$, which is anti-isomorphic to the Lie algebra $\g$:
$$
[\eta_i, \xi_j] = 0,\quad
[\eta_i, \eta_j] =  - C_{ij}^k\, \eta_k.
$$
Taking into account that the vector fields $- r_i$ are infinitesimal generators of the left translations, we obtain
$$
\left ( \eta_i \varphi \right )(\tau_z x_0) = \frac{d}{dt}\, \varphi \left ( \tau_{\exp(t e_i)z}\, x_0 \right ) \big|_{t = 0} = \frac{d}{dt}\, \varphi \left ( \tau_{\exp(t e_i)}\, \tau_z x_0 \right ) \big|_{t = 0},
$$
for any function $\varphi \in C^\infty(M)$. 
From this, it is clear that the vector fields $- \eta_i$ are the infinitesimal generators of $G$ acting on $M$. 
Since the group $G$ acts on $M$ by isometries, we conclude that $- \eta_i$ are the \textit{Killing vector fields} of the pseudo-Riemannian manifold $(M,g)$.

Using~\eqref{eq:06} and the invariance of the metric under the group $G$, we obtain
\begin{equation}
\label{eq:10}
g(\xi_i, \xi_j) = \G_{ij},
\end{equation}
where the constants $\G_{ij}$ are defined as
\begin{equation*}
\label{eq:11}
\G_{ij} \eqdef g((\psi_{x_0})_*\, e_i, (\psi_{x_0})_*\, e_j).
\end{equation*}

It follows from \eqref{eq:08} that the vector fields $\xi_i$ form a basis of the tangent space $T_x M$ at each point $x$ of $M$. 
Consider the collection of 1-forms $\omega^i$ on $M$ uniquely determined by the vector fields $\xi_i$: $\langle \omega^i, \xi_j \rangle = \delta^i_j$. 
These 1-forms form a basis of the cotangent space $T^*_x M$ at $x \in M$ called the \textit{dual} to the basis $\{ \xi_i \}$. 
Using \eqref{eq:10}, we can express the metric $g$ in terms of the 1-form $\omega^i$ as follows:
\begin{equation}
\label{eq:50}
g = \mathbf{G}_{ij}\, \omega^i \omega^j.
\end{equation}
If $(x^1, \dots, x^n)$ is a local coordinate system on~$M$, then for the covariant and contravariant components of the metric we obtain
\begin{equation}
\label{eq:11}
g_{ij}(x) = \G_{kl}\, \omega^k_i(x) \omega^l_j(x),\quad
g^{ij}(x) = \G^{kl}\, \xi_k^i(x) \xi_l^j(x).
\end{equation}
Here $\omega^k_i(x)$ and $\xi_k^i(x)$ are the coordinate components of the 1-form $\omega^k$ and the vector field $\xi_k$, respectively, $\mathbf{G}^{ij} \mathbf{G}_{jk} = \delta^i_k$.

It is immediately clear from the above results that the Hamilton-Jacobi equation \eqref{eq:01} on $(M,g)$ can be expressed in terms of the invariant vector fields $\xi_i$ as
\begin{equation}
\label{eq:12}
\G^{ij} \left ( \xi_i S \right ) \left ( \xi_j S \right )  = m^2.
\end{equation}
Here $\xi_i S$ denotes the directional derivative of a function $S \in C^\infty(M)$ along the $G$-invariant vector field~$\xi_i$.

\section{Canonical coordinates on coadjoint orbits}
\label{sec:2}

The method for constructing a complete integral of the Hamilton--Jacobi equation \eqref{eq:12} developed below substantially uses canonical coordinates on coadjoint orbits of the group of motion~$G$.  In the present section, we recall necessary definitions and describe a convenient algebraic method to construct such coordinates, following the papers~\cite{Shi00,KamPer85}. It should be noted that the idea that the canonical coordinates on coadjoint orbits of Lie groups may be used to get a complete integral of the Hamilton--Jacobi equations is not new (see, for example, \cite{AdaHarHur93}). Here, however, we consider a special class of canonical coordinates related to real polarizations of coadjoint orbits. It turns out that these special coordinates are extremely efficient for reconstructing a complete integral of the Hamilton--Jacobi equation~\eqref{eq:12}.

For any $z \in G$, the automorphism $R_{z^{-1}} \circ L_z: G \to G$ leaves the identity element $e \in G$ fixed. 
Its differential at $e$ is a linear map of the Lie algebra $\mathfrak{g} \simeq T_e G$ into itself. 
This map is denoted by $\mathrm{Ad}_z \eqdef (R_{z^{-1}})_* (L_z)_*$ and called the \textit{adjoint representation} of the group~$G$. 
The \textit{coadjoint representation} of~$G$ is a representation on the dual space $\mathfrak{g}^*$, that is dual to the adjoint representation: $\mathrm{Ad}^*_z \eqdef (\mathrm{Ad}_{z^{-1}})^*$. 
More explicitly, the coadjoint representation is defined by the formula
\begin{equation}
\label{eq:13}
\langle \mathrm{Ad}^*_z\, \lambda, Z \rangle = \langle \lambda, \mathrm{Ad}_{z^{-1}}\, Z \rangle, 
\end{equation}
where $\lambda \in \mathfrak{g}^*$, $Z \in \mathfrak{g}$, and $\langle \cdot, \cdot \rangle$ denotes the natural pairing between the spaces $\mathfrak{g}$ and~$\mathfrak{g}^*$.

The dual space $\mathfrak{g}^*$ of the Lie algebra $\mathfrak{g}$ admits a natural Poisson structure. 
The corresponding Poisson bracket, called the \textit{Lie--Poisson bracket}, has the form 
\begin{equation}
\label{eq:14}
\{ \varphi, \psi \}(f) = C_{ij}^k\, f_k\, \frac{\p \varphi(f)}{\p f_i}\,\frac{\p \psi(f)}{\p f_j},\quad
\varphi, \psi \in C^\infty(\mathfrak{g}^*),
\end{equation}
where $(f_1, \dots, f_n)$ are the coordinates of $f \in \mathfrak{g}^*$ with respect to the basis $\{ e^i \}$ that is dual to the basis $\{ e_i \}$ of~$\mathfrak{g}$. 
In general, the Lie--Poisson bracket~\eqref{eq:14} is degenerate. 
This means that the dual space $\mathfrak{g}^*$ admits a stratification by the symplectic leaves of the Lie--Poisson bracket. 
It has been shown by A.~A.~Kirillov \cite{Kir76} and B.~Kostant~\cite{Kos70} that the symplectic leaves are exactly the coadjoint orbits of the group $G$. 
Thus, the above-mentioned stratification is, in fact, the decomposition of $\g^*$ into coadjoint orbits.

Let $\mathcal{O}_\lambda$ be the coadjoint orbit passing through $\lambda \in \mathfrak{g}^*$. 
From the above, the restriction of the Lie--Poisson bracket to the orbit $\O_\lambda$ is non-degenerate and therefore defines the symplectic form $\omega_\lambda$, called the \textit{Kirillov--Kostant form}, on it. 
By construction, the form $\omega_\lambda$ is invariant under the coadjoint action of $G$ and hence each coadjoint orbit possesses a canonical $G$-invariant symplectic structure.

It follows from Darboux's theorem that there exist local coordinates  $(p_a, q^a)$ on $\mathcal{O}_\lambda$ in which the Kirillov–Kostant form $\omega_\lambda$ takes the form
$$
\omega_\lambda = dp_a \wedge dq^a,\quad
\alpha = 1, \dots, \frac{1}{2}\, \dim \mathcal{O}_\lambda.
$$
The coordinates $(p_a, q^a)$ are called the \textit{canonical coordinates}. 
It is easy to see that the construction of canonical coordinates on the orbit $\mathcal{O}_\lambda$ is reduced to the problem of finding functions $f_i(q,p; \lambda)$, $i = 1, \dots, n$, such that
\begin{equation}
\label{eq:16(a)}
f_i(0, 0; \lambda) = \lambda_i,
\end{equation}
\begin{equation}
\label{eq:16}
\frac{\p f_i(q,p; \lambda)}{\p p_a}\, \frac{\p f_j(q,p; \lambda)}{\p q^a} - 
\frac{\p f_i(q,p; \lambda)}{\p q^a}\, \frac{\p f_j(q,p; \lambda)}{\p p_a} = 
C_{ij}^k\, f_k(q,p; \lambda),
\end{equation}
and
\begin{equation}
\label{eq:17}
\mathrm{rank}\, \left \| \frac{\p f_i(p,q; \lambda)}{\p q^a}, \frac{\p f_i(p,q; \lambda)}{\p p_a} \right \| = \frac{1}{2} \dim \mathcal{O}_\lambda.
\end{equation}

Let us distinguish a special class of canonical coordinates whose functions $f_i(q,p; \lambda)$ are linear in ``momentum''\ variables~$p_a$:
\begin{equation}
\label{eq:18}
f_i(q, p; \lambda) = \zeta_i^a(q) p_a + \chi_i(q; \lambda).
\end{equation}
In this case, the condition~\eqref{eq:16} can be rewritten in the form 
\begin{equation}
\label{eq:19}
\zeta_i^a(q)\, \frac{\p \zeta_j^b(q)}{\p q^a} - 
\zeta_j^a(q)\, \frac{\p \zeta_i^b(q)}{\p q^a} = 
C_{ij}^k\, \zeta_k^b(q),
\end{equation}
\begin{equation}
\label{eq:20}
\zeta_i^a(q)\, \frac{\p \chi_j(q; \lambda)}{\p q^a} - 
\zeta_j^a(q)\, \frac{\p \chi_i(q; \lambda)}{\p q^a} = 
C_{ij}^k\, \chi_k(q; \lambda).
\end{equation}
Furthermore, it follows from  \eqref{eq:17} that
\begin{equation}
\label{eq:19(a)}
\mathrm{rank}\, \| \zeta_i^a(q) \| = \frac{1}{2}\, \dim \mathcal{O}_\lambda.
\end{equation}

In the paper \cite{Shi00}, it is shown that the system of equations~\eqref{eq:19} has a solution satisfying~\eqref{eq:19(a)} if and only if there exists a subalgebra $\mathfrak{h} \subset \mathfrak{g}$ such that
\begin{equation}
\label{eq:21}
\dim \mathfrak{h} = \dim \mathfrak{g} - \frac{1}{2}\, \dim \mathcal{O}_\lambda.
\end{equation}
In this case, the solutions of~\eqref{eq:19} have the following interpretation: the vector fields  $\zeta_i = \zeta_i^a(q) \p_{q^a}$ are infinitesimal generators of a local transitive action of the group $G$ on some smooth manifold $Q$ of dimension $\dim Q = \dim \mathcal{O}_\lambda/2$. 
It will be convenient to assume that $G$ acts on $Q$ on the right. 
Then $Q$ is diffeomorphic to the quotient space  $H \setminus G$, where $H$ is the stationary subgroup of the point~$q = 0$. 
Also, note that the manifold $Q$ can be interpreted as a $G$-invariant Lagrangian submanifold of the symplectic manifold~$\mathcal{O}_\lambda$.

Further, the system of equations \eqref{eq:20} has a solution if and only if the subalgebra $\mathfrak{h}$ is subordinate to the element~$\lambda$~\cite{Shi00}:
\begin{equation}
\label{eq:22}
\langle \lambda, [\mathfrak{h}, \mathfrak{h}] \rangle = 0.
\end{equation}
A subalgebra $\mathfrak{h} \subset \mathfrak{g}$ satisfying \eqref{eq:21} and \eqref{eq:22} is called \textit{polarization} of the element~$\lambda \in \mathfrak{g}^*$. 
Thus, the canonical coordinates on $\mathcal{O}_\lambda$ defined by the transition functions \eqref{eq:18} exist if and only if the element $\lambda$ has a polarization.

Currently, the problem of the existence of polarizations in Lie algebras has been well studied. 
A quite complete list of results on this problem is given in Dixmier's book~\cite{Dix77}. 
In particularly, polarizations always exist for nilpotent and completely solvable Lie algebras.
On the other hand, for a given semisimple Lie algebra $\g$, not every element from $\g^*$ admits a polarization. 
But, in spite of this fact, the canonical transition \eqref{eq:18} can still be constructed in this case as well. 
This requires the consideration of polarizations in the complexifications of Lie algebras (see \cite{Shi00} for details).

We note that if a polarization $\mathfrak{h}$ of $\lambda \in \mathfrak{g}^*$ exists, then the problem of constructing canonical coordinates on the coadjoint orbit $\mathcal{O}_\lambda$ can be constructively solved by the tools of linear algebra. 
Indeed, in the paper \cite{MagMikShi15}, it is shown that infinitesimal generators of $G$ acting on a homogeneous space $Q = H \setminus G$ can be constructed from the structure constants of Lie group $G$ by the computation of matrix inversions and matrix exponents. 
The functions $\chi_i(q; \lambda)$, which satisfy~\eqref{eq:20}, can also be found without direct solving the differential equations. 
An algebraic method of constructing such functions is described in the work~\cite{BarShi09}.

\section{Constructing a complete integral of the Hamilton--Jacobi equation}
\label{sec:3}

Now we come to the question of how to construct a complete integral for the geodesic Hamilton--Jacobi equation~\eqref{eq:12}. 
It turns out that this problem can be reduced to finding a complete integral for Hamilton–Jacobi equations on Lagrangian submanifolds of regular coadjoint orbits of~$G$. 
Thus, solving the original differential equation with $n = \dim G$ independent variables can be reduced to solving some auxiliary differential equation in which the number of independent variables equals $r = \frac{1}{2} \, \dim \mathcal{O}_\lambda < n$.

Before formulating the basic result, we introduce some additional constructions.

Let $G$ be a Lie group and $H \subset G$ be its connected closed subgroup. 
Denote by $\mathfrak{g}$ and $\mathfrak{h}$ the Lie algebras of the groups $G$ and $H$, respectively. 
Let us consider the right homogeneous space $Q = H \setminus G$. 
The group $G$ naturally acts on $Q$; we denote this action as $\rho: G \times Q \to Q$, $q \mapsto \rho_z q$, $q \in Q$, $z \in G$. 
To each basis vector $e_i$ of the Lie algebra $\mathfrak{g}$, we can associate the infinitesimal generator $\zeta_i$ of the action of $G$ on $Q$ by the formula
$$
\label{eq:27}
\left ( \zeta_i\, f \right )(q) \eqdef \frac{d}{dt}\, f \left ( \rho_{\exp(t e_i)}\, q \right ) \Big|_{t = 0},\quad
f \in C^\infty(Q).
$$

Let us define the map $\varphi: M \times Q \to Q$ by
\begin{equation}
\label{eq:25}
\varphi(x,q) \eqdef \rho_{\psi_{x_0}^{-1}(x)}\, q,\quad
x \in M,\quad
q \in Q,
\end{equation}
where $\psi_{x_0}^{-1}: M \to G$ is the inverse of the diffeomorphism~\eqref{eq:05}. 
Clearly, $\varphi(x_0, q) = q$. 
Moreover, from the equality $\rho_{z_1 z_2} = \rho_{z_2} \rho_{z_1}$ it follows that
\begin{equation}
\label{eq:26}
\rho_z\, \varphi(x, q) = \varphi \left ( \psi_{x_0} (\psi_{x_0}^{-1}(x) z), q \right ),
\end{equation} 
for all $q \in Q$, $x \in M$, and $z \in G$. 
Let $U \subset M$ and $V \subset Q$ be the domains such that $\varphi(U \times V) \subset V$ and let $x^1, \dots, x^n$ and $q^1, \dots, q^r$ be the systems of local coordinates in $U$ and $V$, respectively. 
Denote by $\varphi^a(x,q) = \varphi^a(x^1, \dots, x^n, q^1, \dots, q^r)$ the collection of functions defining the map $\varphi(x,q)$ in the local coordinates, $a = 1, \dots, r$. 
Then, from \eqref{eq:26} and the definition of the vector field $\xi_i$, it follows that
\begin{equation}
\label{eq:28}
 \xi_i^j(x)\, \frac{\p \varphi^a(x,q)}{\p x^j} = \zeta_i^a(\varphi(x,q)).
\end{equation}
Here $\xi_i^j(x)$ and $\zeta_i^a(q)$ are the coordinate components of the vector fields $\xi_i$ and $\zeta_i$, respectively.

Note that the above relations are correct for any subalgebra $\mathfrak{h} \subset \mathfrak{g}$. 
Next, we apply these results to the case when $\mathfrak{h}$ is a polarization of an element~$\lambda \in \mathfrak{g}^*$. 
As we saw earlier, the homogeneous space $Q = H \setminus G$, in this case, is an invariant Lagrangian submanifold of the coadjoint orbit~$\mathcal{O}_\lambda$.

Before proceeding, we recall some terminology concerning coadjoint orbits.

An element $\lambda \in \mathfrak{g}^*$ is called \textit{regular} if the coadjoint orbit $\mathcal{O}_\lambda$ passing through $\lambda$ has the maximal dimension in~$\mathfrak{g}^*$. 
For a regular element $\lambda \in \mathfrak{g}^*$, the dimension of $\mathcal{O}_\lambda$ can be calculated by the formula
$$
\dim \mathcal{O}_\lambda = \dim \mathfrak{g} - \mathrm{ind}\, \mathfrak{g},
$$
where the non-negative integer $\mathrm{ind}\, \mathfrak{g}$, called the \textit{index} of the Lie algebra $\mathfrak{g}$, is defined as
\begin{equation}
\label{eq:24}
\mathrm{ind}\, \mathfrak{g} \eqdef \inf \limits_{\lambda \in \mathfrak{g}^*} \mathrm{corank}\, \| C_{ij}^k\, \lambda_k \|.
\end{equation}
For a semisimple Lie algebra $\mathfrak{g}$, the index $\mathrm{ind}\, \mathfrak{g}$ coincides with its rank.

Let us fix a regular element $\lambda_0 \in \mathfrak{g}^*$. 
Since $\lambda_0$ is in general position, the coadjoint orbits close to $\mathcal{O}_{\lambda_0}$ are diffeomorphic to it. 
Thus, there exists a small neighborhood $U \subset \mathfrak{g}^*$ of $\lambda_0$ that is stratified on the homomorphic $G$-fibers. 
Consider the quotient space $J = U / G$, which is the base space of this fibration. 
It is clear that $\dim J = \mathrm{ind}\, \mathfrak{g}$. 
Let $j = (j_1, \dots, j_{\mathrm{ind}\, \mathfrak{g}})$ be the local coordinates on the manifold $J$ and denote by $\lambda(j)$ some smooth local section of the fibration $U \to U / G$.
Then the correspondence $j \to \mathcal{O}_{\lambda(j)}$ defines some smooth one-to-one parametrization of the coadjoint orbits in the neighborhood~$U$.

Now we formulate the main result.

\begin{theorem}
\label{the:01}
Let $\lambda(j) = \lambda(j_1, \dots, j_{\mathrm{ind}\, \mathfrak{g}})$ be a smooth (local) parametrization of regular coadjoint orbits in $\mathfrak{g}^*$, $\mathfrak{h} \subset \mathfrak{g}$ be a polarization of the element $\lambda(j)$, and $f_i(q,p; \lambda(j)) = \zeta_i^a(q) p_a + \chi_i(q; \lambda(j))$ be the functions that define the transition to canonical coordinates on $\mathcal{O}_{\lambda(j)}$. 
Let us assume that there exists the closed subgroup $H \subset G$ with the Lie algebra $\mathfrak{h}$. 
On the homogeneous space $Q = H \setminus G$, we consider the differential equation
\begin{equation}
\label{eq:29}
\mathbf{G}^{ij} f_i \left ( q, \frac{\p \tilde{S}}{\p q}; \lambda(j) \right ) f_j \left ( q, \frac{\p \tilde{S}}{\p q}; \lambda(j) \right ) = m^2,
\end{equation}
and denote by $\tilde{S}_j(q; \beta)$ its complete integral depending on a set of parameters $\beta = (\beta_1, \dots, \beta_{(\dim \mathfrak{g} - \mathrm{ind}\, \mathfrak{g})/2})$. 
Then the function
\begin{equation}
\label{eq:30}
S(x;\alpha) = \tilde{S}_j(\varphi(x,q); \beta) + \int \chi_k(\varphi(x,q); \lambda(j)) \omega^k(x)
\end{equation}
is a complete integral of the Hamilton--Jacobi equation~\eqref{eq:12}. 
Here $\alpha = (q, j, \beta)$ is the set of $n$ parameters, $\omega^k(x)$ are 1-forms dual to the vector fields $\xi_k(x)$, and the mapping $\varphi: M \times Q \to Q$ is defined by the formula~\eqref{eq:25}.
\end{theorem}

\begin{proof}
First we show that \eqref{eq:30} satisfies the Hamilton--Jacobi equation ~\eqref{eq:12}. 
From Eq.~\eqref{eq:28} and the fact that the 1-forms $\omega^k$ are dual to the vector fields $\xi_i$, we obtain 
\begin{multline}
\label{eq:31}
\xi_i^k(x)\, \frac{\p S(x; \alpha)}{\p x^k} = \xi_i^k(x)\, \frac{\p \varphi^a(x,q)}{\p x^k}\, \frac{\p \tilde{S}_j(\varphi(x,q); \beta)}{\p \varphi^a(x,q)} + \chi_k(\varphi(x,q); \lambda(j)) \langle \omega^k, \xi_i \rangle =
\\
= \zeta_i^a(\varphi(x,q))\, \frac{\p \tilde{S}_j(\varphi(x,q); \beta)}{\p \varphi^a(x,q)} + \chi_i(\varphi(x,q); \lambda(j)) = 
\\
= f_i \left ( q', \frac{\p \tilde{S}_j(q'; \beta)}{\p q'}; \lambda(j) \right ) \Big|_{q' = \varphi(x,q)}.
\end{multline}
Since $\tilde{S}_j(q; \beta)$ is a solution of~\eqref{eq:29}, it is clear that \eqref{eq:30} satisfies~\eqref{eq:12}.

Now we prove that the function \eqref{eq:30} obeys the condition~\eqref{eq:02}. 
We introduce the notation
$$
q'^a = \varphi(x,q),\quad
p'_a = \frac{\p \tilde{S}_j(q';\beta)}{\p q'^a},\quad
u = (q', p', J).
$$
Using~\eqref{eq:31}, we obtain
$$
\frac{\p^2 S(x; \alpha)}{\p x^i \p \alpha_j} = \omega^k_i(x)\, \frac{\p f_k(q', p', \lambda(j))}{\p \alpha_j} = \omega^k_i(x)\, \frac{\p f_k(u)}{\p u^l}\, \frac{\p u^l}{\p \alpha_j},
$$
whence
\begin{equation}
\label{eq:34}
\det \left \| \frac{\p^2 S(x;\alpha)}{\p x^i \p \alpha_j} \right \| = \det \| \omega^k_i(x) \| \cdot \det \left \| \frac{\p f_k(u)}{\p u^l} \right \| \cdot \det \left \| \frac{\p u^l}{\p \alpha_j} \right \|.
\end{equation}
Clearly, the first determinant on the right-hand side of the last equality is non-zero since the collection of 1-forms $\{ \omega^k(x) \}$ forms a basis in $T^*_x M$ at any $x \in M$. 
The second factor on the right-hand side of~\eqref{eq:34} is also non-zero, because, by construction, the functions $f_i(q,p;\lambda)$ define a local immersion of orbit $\mathcal{O}_\lambda$ in the dual space~$\mathfrak{g}^*$. 
In order to show that the third determinant cannot be zero, we rewrite it in the form
$$
\det \left \| \frac{\p u^l}{\p \alpha_j} \right \| = 
\left |
\begin{array}{ccc}
\frac{\p \varphi(x,q)}{\p q}	&	0	&	0	\\
0								&	I	&	0	\\
*								&	*	&	\frac{\p^2 \tilde{S}_j(q;\beta)}{\p q\, \p \beta}
\end{array}
\right |,
$$
From this, we obtain
$$
\det \left \| \frac{\p u^l}{\p \alpha_j} \right \| = \det \left \| \frac{\p \varphi^a(x,q)}{\p q^b} \right \| \cdot 1 \cdot \det \left \| \frac{\p^2 \tilde{S}_j(q;\beta)}{\p q^a\, \p \beta_b}\right \|.
$$
Since the function $\tilde{S}_j(q; \beta)$ is a complete integral of \eqref{eq:29}, the right hand side of the last equality does not vanish by the definition of the map~$\varphi$.
\end{proof}

The above theorem allows us to reduce the problem of construction of a complete integral of the Hamilton--Jacobi equation \eqref{eq:12} to the problem of finding a complete integral of subsidiary Hamilton--Jacobi equations on Lagrangian submanifolds of regular coadjoint orbits. 
In particular, if the dimension of regular coadjoint orbits is less than or equal to 2, then a complete integral of~\eqref{eq:12} can be found in quadratures. 
Indeed, if we use the equality
$$
\mathbf{G}^{ij} f_i(q,p; \lambda(j)) f_j(q,p; \lambda(j)) = m^2,
$$
we can express the variable $p$ as a function of $q$, $j$, and $\beta = m$: $p = p(q; j, m)$. 
Hence for the function $\tilde{S}_j(q; m)$, we have
$$
\tilde{S}_j(q; m) = \int p(q; j, m)\, dq.
$$
Substituting this function into~\eqref{eq:30}, we obtain a complete integral of the Hamilton--Jacobi equation~\eqref{eq:12}.
We thus have proved the following consequence of Theorem~\ref{the:01}.

\begin{consequence}
\label{con:1}
Let $(M,g)$ be a pseudo-Riemannian manifold with a simply transitive group of motion $G$. If the dimension of regular coadjoint orbits of $G$ is less than or equal to 2, then a complete integral of the geodesic Hamilton--Jacobi equation on $(M,g)$ can be found in quadratures.
\end{consequence}

\section{An example: a complete integral for the geodesic Hamilton--Jacobi equation in the McLenaghan--Tariq--Tupper spacetime}
\label{sec:4}

As an example of the application of our technique, let us consider the problem of constructing a complete integral of the geodesic Hamilton--Jacobi equation in the \textit{McLenaghan--Tariq--Tapper spacetime}. McLenaghan and Tariq found a solution to the Einstein--Maxwell equations whose electromagnetic tensor does not share the spacetime symmetry~\cite{LenTar75}.
The line element of this metric can be written as
\begin{multline}
\label{eq:41}
d s^2 = (d x^1)^2 + 2 d x^1 d x^2 - 2 k x^3 ( d x^1 + d x^2 ) d x^4 - 
\\
 - e^{- k x^2} d (x^3)^2 + \left [ k^2 (x^3)^2 - e^{k x^2} \right ] (d x^4)^2, 
\end{multline}
where $k$ is an arbitrary positive real number;  the special case of this metric corresponding to $k = 4$ was found by Tariq and Tupper~\cite{TarTap75}.

As it was shown in~\cite{LenTar75}, the group of motions $G$ of the metric \eqref{eq:41} is generated by the Killing vectors
$$
\eta_1 = \p_{x^1},\
\eta_2 = \p_{x^4},\
\eta_3 = k x^4 \p_{x^1} + \p_{x^3},\
\eta_4 = - 2 \p_{x^1} + \p_{x^2} + \frac{k}{2} \left ( x^3 \p_{x^3} - x^4 \p_{x^4} \right ),
$$
and is simply transitive by virtue of the condition $\det \| \eta_i^j(x) \| \neq 0$.
The corresponding group action $\tilde{x} = \tau_z(x)$, expressed in the local coordinates $x^i$, takes the form 
$$
\tilde{x}^1 = x^1 - z^1 + 2 \, z^4 - k e^{k z^4/2} z^3 x^4,\quad
\tilde{x}^2 = x^2 - z^4,\quad
\tilde{x}^3 = x^3 e^{-k z^4/2} - z^3,
$$
$$
\tilde{x}^4 = x^4 e^{k z^4/2} - z^2.
$$
Here, we denote by $z = (z^1, z^2, z^3, z^4)$ the group parameters.

Let us fix the point $x_0 = (0,0,0,0)$; then the mapping $\psi_{x_0}(z)$ defined by~\eqref{eq:05} can be written as
\begin{equation}
\label{eq:45}
\psi_{x_0}^1(z) = 2 z^4 - z^1,\quad
\psi_{x_0}^2(z) = -z^4,\quad
\psi_{x_0}^3(z) = -z^3,\quad
\psi_{x_0}^4(z) = -z^2.
\end{equation}
Using \eqref{eq:08}, we obtain the following expressions for the invariant vector fields~$\xi_i$:
\begin{equation}
\label{eq:47}
\xi_1 = - \p_{x^1},\
\xi_2 = - e^{ - \frac{k x^2}{2}} \left ( k x^3 \p_{x^1} + \p_{x^4} \right ),\
\xi_3 = - e^{\frac{k x^2}{2}} \p_{x^3},\
\xi_4 = 2 \p_{x^1} - \p_{x^2}.
\end{equation}
These vector fields generate the Lie algebra $\mathfrak{g}$ of the group $G$ with the commutation relations
$$
[\xi_1, \xi_2] = [\xi_1, \xi_3] = [\xi_1, \xi_4] = 0,\
[\xi_2, \xi_3] = k \xi_1,\
[\xi_2, \xi_4] = - \frac{k}{2} \, \xi_2,\
[\xi_3, \xi_4] = \frac{k}{2} \, \xi_3.
$$
It is easy to see that $\mathfrak{g}$ is a one-dimensional central extension of the three-dimensional algebra $\langle \xi_2, \xi_3, \xi_4 \rangle$ of the Bianchi type VI$_0$.

In the tetrad basis \eqref{eq:47}, the McLenaghan--Tariq--Tupper metric takes the form
$$
\| \mathbf{G}_{ij} \| = 
\left (
\begin{array}{rrrr}
1	&	0	&	0	&	-1	\\
0	&	-1	&	0	&	0	\\
0	&	0	&	-1	&	0	\\
-1	&	0	&	0	&	0
\end{array}
\right ).
$$
Then, in accordance with \eqref{eq:50}, we have
\begin{equation}
\label{eq:51}
ds^2 = \mathbf{G}_{ij} \omega^i \omega^j = (\omega^1)^2 - 2 \, \omega^1 \omega^4 - (\omega^2)^2 - (\omega^3)^2,
\end{equation}
where the one-forms $\omega^i$ are dual to the vector fields $\xi_i$:
\begin{equation}
\label{eq:52}
\omega^1 = - d x^1 - 2 d x^2 + k x^3 d x^4,\
\omega^2 = - e^{\frac{k x^2}{2}} d x^4,\
\omega^3 = - e^{- \frac{k x^2}{2}} d x^3,\
\omega^4 = - d x^2.
\end{equation}
It follows from \eqref{eq:51} that the geodesic Hamilton--Jacobi equation \eqref{eq:01} for the McLenaghan--Tariq--Tupper metric can be written as
$$
2 (\xi_1 S) (\xi_4 S) + (\xi_2 S)^2 + (\xi_3 S)^2 + (\xi_4 S)^2 + m^2 = 0,
$$
or, in the explicit form,
\begin{multline}
\label{eq:54}
2 \left ( \frac{\p S}{\p x^1} \right ) \left ( \frac{\p S}{\p x^2} - 2 \, \frac{\p S}{\p x^1} \right ) + e^{- k x^2} \left ( k x^3 \, \frac{\p S}{\p x^1} + \frac{\p S}{\p x^4} \right )^2
 + 
 \\
 + e^{k x^2} \left ( \frac{\p S}{\p x^3} \right )^2 + \left ( \frac{\p S}{\p x^2} - 2 \, \frac{\p S}{\p x^1} \right )^2 + m^2 = 0.
\end{multline}
It is important to note that the Hamilton--Jacobi equation \eqref{eq:54} has non-St\"ackel form, i.e. it cannot be solved by the separation of variables on the configuration space $M$.
In order to make sure this, it is enough to verify that the metric \eqref{eq:41} does not satisfy the necessary and sufficient conditions of separability obtained by V.~N.~Shapovalov~\cite{Sha79}.

In order to construct a complete integral of the Hamilton--Jacobi equation \eqref{eq:54}, we apply the method outlined in Sec.~\ref{sec:3}. 
We note that this problem can be solved by quadratures since it follows from~\eqref{eq:24} that $\mathrm{ind}\, \mathfrak{g} = 2$; therefore the dimension of regular coadjoint orbits of $G$ equals $\dim \mathfrak{g} - \mathrm{ind}\, \mathfrak{g} = 2$ (see Consequence~\ref{con:1}).

The matrix $\| \mathrm{Ad}_{z^{-1}} \|$ of the adjoint representation can be expressed in terms of the Killing vector fields $\eta_i$ and the invariant 1-forms $\omega^i$ as follows
$$
\| \mathrm{Ad}_{z^{-1}} \|^i_j = -\omega^i_k(x) \eta^k_j(x) \big|_{x = \psi_{x_0}(z)},
$$
so that
$$
\| \mathrm{Ad}_{z^{-1}} \| = \left (
\begin{array}{cccc}
1	&	k z^3		&	- k z^2		&	k^2 z^2 z^3/2	\\
0	&	e^{-k z^4/2}&	0			&	k z^2 e^{-k z^4/2}/2		\\
0	& 	0			&	e^{ k z^4/2}&	- k z^3 e^{k z^4/2}/2		\\
0	&	0			&	0		&	1
\end{array}\right ).
$$
In accordance with \eqref{eq:13}, the action of $\mathrm{Ad}^*_z$ on $f = (f_1, f_2, f_3, f_4) \in \mathfrak{g}^*$ takes the form
\begin{equation}
\label{eq:56(a)}
(\mathrm{Ad}_z^*\, f)_1 = f_1,\
(\mathrm{Ad}_z^*\, f)_2 = k z^3 f_1 + e^{- k z^4/2} f_2,\
(\mathrm{Ad}_z^*\, f)_3 = - k z^2 f_1 + e^{k z^4/2} f_3,
\end{equation}
\begin{equation}
\label{eq:56(b)}
(\mathrm{Ad}_z^*\, f)_4 = \frac{k^2}{2}\, z^2 z^3 f_1 + \frac{k}{2}\, z^2 e^{-k z^4/2} f_2 - \frac{k}{2}\, z^3 e^{k z^4/2} f_3 + f_4.
\end{equation}
It is easy to see that the functions
\begin{equation*}
\label{eq:57(a)}
K_1 = f_1,\quad
K_2 = 2 f_1 f_4 + f_2 f_3
\end{equation*}
are functionally independent invariants of this action; therefore the connected components of their level sets are coadjoint orbits of $G$.

Let us introduce a local parametrization of the regular coadjoint orbits: $\lambda(j) = (j_1, 0, 0, j_2)$, $j_1 \neq 0$. 
It is easy to verify that the subalgebra $\mathfrak{h} = \langle \xi_1, \xi_3, \xi_4 \rangle \subset \mathfrak{g}$ is a polarization of $\lambda(j) \in \mathfrak{g}^*$. 
The corresponding canonical coordinates $(p, q)$ on the coadjoint orbit passing through the element $\lambda(j)$ are defined by the relations
\begin{equation}
\label{eq:58}
f_1 = j_1,\quad
f_2 = p,\quad
f_3 = k j_1 q,\quad
f_4 = j_2 - \frac{k}{2}\, p q.
\end{equation}
Here $q$ is a local coordinate on the homogeneous space $Q = \exp(\mathfrak{h}) \setminus G$, which is $G$-invariant Lagrangian submanifold of the orbit~$\mathcal{O}_{\lambda(j)}$.

From \eqref{eq:58}, we obtain that $q = f_3/(k j_1)$. 
Using \eqref{eq:56(a)} and \eqref{eq:56(b)}, we can easy reconstruct the action $\rho:G \times Q \to Q$:
$$
\rho_z(q) = \frac{(\mathrm{Ad}^*_{z^{-1}} f)_3}{k j_1} =  e^{- k z^4/2} \left ( q + z^2 \right ).
$$
From \eqref{eq:45} we get the following expression for the function $\varphi(x,q)$ (see the formula \eqref{eq:25}):
\begin{equation}
\label{eq:60}
\varphi(x,q) = e^{k x^2/2} \left ( q - x^4 \right ).
\end{equation}

The equation \eqref{eq:29} for the function $\tilde{S}_j = \tilde{S}_j \left ( q; m \right )$  can be written by using the explicit form of the functions $f_i(q, p; \lambda)$:
\begin{equation}
\label{eq:61}
\left ( 1 + \frac{k^2 q^2}{4} \right )^2 \left ( \frac{d \tilde{S}_j}{d q} \right )^2 - 
k q \left ( j_1 + j_2 \right ) \, \frac{d \tilde{S}_j}{d q} + k^2 j_1^2 q^2 + 2 j_1 j_2 + j_2^2 + m^2 = 0.
\end{equation}
Expressing the derivative of $\tilde{S}_j(q; m)$ with respect to $q$, we obtain
\begin{equation}
\label{eq:62}
\frac{d \tilde{S}_j(q; m)}{d q} = \frac{2 k (j_1 + j_2) q + 2 \sqrt{D(k^2 q^2; j, m)}}{4 + k^2 q^2},
\end{equation}
where
$$
D(\theta; j, m) = - j_1^2 \theta^2 - ( m^2 + 3 j_1^2) \theta - 4 ( m^2 + 2 j_1 j_2 + j_2^2 ).
$$
The equation \eqref{eq:62} has a real solution in some open set $U \subset \mathbb{R}$ if and only if 
$$
\Delta \geq 0,\quad
\sqrt{\Delta} > m^2 + 3 j_1^2,
$$
where $\Delta \eqdef ( m^2 - j_1^2 + 4 j_1 j_2 )( m^2 - 9 j_1^2 - 4 j_1 j_2 )$ is the discriminant of the quadratic polynomial $D(\theta; j, E)$.
In this case, the quadratic equation $D(\theta; j, E) = 0$ has the real roots
\begin{equation}
\label{eq:65}
\theta_+ = \frac{\sqrt{\Delta} - ( m^2 + 3 j_1^2)}{2 j_1^2} > 0,\quad
\theta_- = - \frac{\sqrt{\Delta} + ( m^2 + 3 j_1^2)}{2 j_1^2} < 0,
\end{equation}
and the domain $U$ is defined as
$$
U = \{ q \in \mathbb{R} \colon - \frac{\sqrt{\theta_+}}{k} < q < \frac{\sqrt{\theta_+}}{k} \}.
$$

The function $\tilde{S}_j(q; m)$ can be found by integration of the right-hand side of \eqref{eq:62} and, after some algebra, can be expressed in terms of incomplete elliptic integrals:
\begin{multline}
\label{eq:67}
\tilde{S}_j(q; m) = \int \limits_0^{q} \frac{2 k (j_1 + j_2) q + 2 \sqrt{D(k^2 q^2; j, m)}}{4 + k^2 q^2} \, d q = 
\\
= \frac{j_1 + j_2}{k} \, \ln \left ( 1 + \frac{k^2 q^2}{4} \right ) + \frac{2 j_1}{k} \left [  \sqrt{\theta_+ - \theta_-} \, \mathrm{E} \left ( \sqrt{1 - \frac{k^2 q^2}{\theta_+}}, \frac{\theta_+}{\theta_+ - \theta_-} \right ) - \right .
\\
\left . - \frac{4 + \theta_+}{\sqrt{\theta_+ - \theta_-}} \, \mathrm{F} \left ( \sqrt{1 - \frac{k^2 q^2}{\theta_+}}, \frac{\theta_+}{\theta_+ - \theta_-} \right ) +
 \right . \\ \left . + 
 \frac{4 + \theta_-}{\sqrt{\theta_+ - \theta_-}} \, \Pi \left ( \sqrt{1 - \frac{k^2 q^2}{\theta_+}}, \frac{\theta_+}{4 + \theta_+}, \frac{\theta_+}{\theta_+ - \theta_-} \right ) \right ].
\end{multline}
Here the functions $\mathrm{F}(z,\kappa)$, $\mathrm{E}(z,\kappa)$, and $\Pi(z,a,\kappa)$ are incomplete elliptic integrals in the Legendre normal forms of the first, second and third kind, respectively~\cite{GraRyz14}. 
The parameters $\theta_+$ and $\theta_-$ are defined by the relations~\eqref{eq:65}.

A complete integral of the Hamilton--Jacobi equation \eqref{eq:54} is given by the formula \eqref{eq:30}. 
Using \eqref{eq:52}, \eqref{eq:58}, and \eqref{eq:60}, we obtain the following expression for $S(x; \alpha)$:
\begin{equation*}
S(x; \alpha) = \tilde{S}_j \left ( e^{k x^2/2} ( q - x^4 ); m \right ) - j_1 x^1 - ( 2 j_1 + j_2 ) x^2 + k j_1 ( x^4 - q ) x^3.
\end{equation*}
Here the function $\tilde{S}_j \left ( q; m \right )$ is given by \eqref{eq:67}, and $\alpha = (q, j_1, j_2, m)$ is a set of parameters.

\section*{Acknowledgements}

The author expresses his deep gratitude to Igor V.~Shirokov for his constant and fruitful discussions and support.
Dr. S.~V.~Danilova is gratefully acknowledged for careful reading of the manuscript.

\bibliography{mybibfile}

\end{document}